%% file: hash_opti_kangaroov1.5.tex
\documentclass{llncs}
\usepackage[utf8]{inputenc}

\usepackage{stmaryrd}
\usepackage{graphicx}
\usepackage{xcolor}
\usepackage[lofdepth,lotdepth]{subfig}

\renewenvironment{proof}{\paragraph{Proof} }{\hfill\qed}
\renewcommand{\paragraph}[1]{\noindent\textit{#1}.}

\begin{document}

\title{Optimal Tree Hash Modes: the Case of Trees Having their Leaves at All the Levels}

\author{Kevin Atighehchi}
\institute{Aix-Marseille Univ, CNRS, LIF, Marseille, France\\
\email{kevin.atighehchi@univ-amu.fr}\\
}

\maketitle

\begin{abstract}
A recent work shows how  we can optimize a tree based mode of operation for a hash function where the sizes of input 
message blocks and digest are the same, subject to the constraint that the involved tree structure has all its leaves at the same depth.
In this work, we show that we can further optimize the running time of such a mode by using a tree having leaves at all its levels. We make the assumption that
the input message block has a size a multiple of that of the digest and denote by $d$ the ratio block size over digest size. The running time 
is evaluated in terms of number of operations performed by the hash function, \emph{i.e.} the number of calls to its underlying function.
It turns out that a digest can be computed 
in $\lceil \log_{d+1} (l/2) \rceil+2$ evaluations of the underlying function using $\lceil l/2 \rceil$ processors, 
where $l$ is the number of blocks of the message. Other results of interest are discussed, such as the optimization of the parallel running time
for a tree of restricted height.
\end{abstract}


\section{Introduction}

In the cryptographic hashing context, we are interested in the problem of finding a tree structured circuit topology 
to optimize both the parallel running time and the number of involved processors (\emph{i.e.} in time and width).
We consider hash tree modes using a hash function (or variable-input-length compression function), denoted $f$, 
where the ratio block size over digest size is an integer denoted $d$.
For instance, they can correspond to SBL (single-block-length) hash functions based on a block
cipher having the key and the block of the same size.
We consider that the hash function 
needs only $l$ invocations of the underlying primitive to process a $l$-block message. 
Let us assume a hash tree of height $h$ having all its leaves (\emph{i.e.} message blocks) at the same depth. If we denote
by $a_i$ the arity of level $i$ (for $i=1 \ldots h$) and if $d=1$, then the parallel 
running time to obtain the root node value is $\sum_{i=1}^h a_i$. A recent work~\cite{AtRo17,AR15} shows that we can select the good parameters
to construct such trees that minimize both the running time and the number of processors.
The aim of the present paper is to show that we can further decrease the parallel running time of a tree-based hash function
by removing this structural constraint on the tree. We then remark that the allocation of tasks to the processors is a bit more subtle, and 
that the parallel running time is no longer the sum of the level arities.
More particularly, our contributions are the followings:
\begin{itemize}
 \item 
 We first recall that it is possible to design a hash function whose implementation will behave like an 
 idealized hash function from the rate standpoint. In particular, assuming precomputations, 
 this hash function requires $l$ calls to the underlying primitive to process a message of $l$ blocks. 
 This resulting sequential hash function is then used as building block for tree hashing.
 \item We show the parallel running time which can be obtained using hash trees of smallest height. In particular, we state a result
 in which both the running time and the number of involved processors are optimized. A tree of minimized height has the benefit of minimizing the memory consumption.
 \item We then address the case of trees of unrestricted height. We show the optimal parallel running time which can be obtained in this case and discuss 
 how the number of involved processors can be decreased without changing this running time.
 \item We finally consider a situation in which we have a \emph{bounded parallelism}. We show the optimal parallel running time which can be obtained with a fixed number of
 processors.
\end{itemize}

This paper is organized in the following way. We give some definitions about trees and hash functions in Section \ref{terminology} and
we address the optimization of tree constructions suitable for parallel hashing in Section \ref{kangaroo}. 

\section{Terminology and background information}\label{terminology}

Throughout this paper, we use the 
convention\footnote{This corresponds to the convention used to describe Merkle trees. The other (less frequent)
convention is to define a node as being 
an $f$-input.
} that a node is the result of
an inner function $f$ called on a data composed of the node's children.
A node value then corresponds to an image by such a function and a child of this node can be either
an other image or a message block.
In this paper, a $k$-ary tree of height $h$ is a tree having the following properties:
\begin{itemize}
 \item The root node (at level $h$) can be of arity $a$, with $1 < a \leq k$.
 \item A level $i$ ($\neq h$) has all its nodes of arity $k$, except the rightmost one that can be of smaller arity.
\end{itemize}
We define the arity of a level in the tree as being the greatest node arity in this level.~\\

Let us denote the block size and the digest size $N_b$ and $N_o$ respectively. We make the assumption that $d=N_b/N_o$ is a positive integer.
A node in the tree is computed using an inner VIL function that iteratively processes message blocks of size $N_b$ bits using an 
underlying function (a block cipher, a permutation or another compression function) and produces a digest of $N_o$ bits. 
The underlying function is considered
to be the lowest level function.
For instance, the hash function Skein \cite{FLSWBKCW09} is based on a VIL compression function, itself based on a lowest level primitive, 
the tweakable block cipher \emph{Threefish}.

We assume that the evaluation of the inner function requires a number of calls to its underlying function equal to the number of blocks of the message. 
At first sight we could think that this kind of primitive is rare since: ($i$) there is usually a padding which is done at the end of the message. 
For certain message sizes, this padding requires one more call to the underlying function; ($ii$) In the hash functions like SHA-1 and SHA-2, 
the MD-strengthening add another block containing the message size.
However, we show that we can construct an inner VIL function that can satisfy a running time of $l$ unit of times for a message of $l$ blocks.
Besides, some existing inner functions are already of this type, such as the VIL compression function based on the UBI (Unique Block Iteration) chaining mode of 
Skein \cite{FLSWBKCW09} and some other \emph{single-block-length} hash functions \cite{PGV94,KM07}.

In this paper, the time complexity corresponds to the number of 
evaluations of the lowest level function and we use the term \emph{unit of time} for one evaluation of such a function.




\subsection{Computation model}

We use the PRAM (Parallel Random Access Machine) model of computation, assuming the strategy CREW (Concurrent Read Exclusive Write), although 
in the context of hashing, the strategy EREW can be enough in most cases.

Except when otherwise specified, the parallel running time corresponds to the running time when the number of processors is 
not \emph{a priori} bounded. As a consequence, the message is supposed to be already avalaible.
In the hash tree constructions that we propose, if the number of chaining values is denoted $n_{cv}$ and if
the root node is not counted as such, then the number of processors is equal to $n_{cv}+1$. Indeed, the chaining values are computed by distinct processors.


\subsection{Concretizing an idealized rate for the inner function}

According to Bertoni \emph{et al.} \cite{BDPV14_Suf}, the tree-based operating mode 
is indifferentiable from a random oracle if their three conditions are fulfilled and if the operated inner 
function is indifferentiable from a random oracle.~\\

If we want to use a hash function based on the Merkle-Damg{\aa}rd construction, we need to use the modifications proposed 
by Coron \emph{et al.} \cite{CDMP05}, which ensure that the inner function will be indifferentiable from a random oracle.~\\

\paragraph{\textbf{Inner function based on the \emph{prefix-free} MD}} We choose to use a modification of the Merkle-Damg{\aa}rd construction, proposed by 
Coron \textit{et al.} \cite{CDMP05}. Let us denote $M$ a message to hash, 
padded with a bit 1 and the minimum\footnote{possibly 0.} number of bits 0, such that the length of the padded message is a multiple of $N_b$.
The modification consists, before applying the MD algorithm, in prepending to this padded message a block containing the length of $M$ in bits. 
We denote by $f'$ the resulting hash function. Coron \textit{et al.} show that $f'$ is indifferentiable from a random oracle.

Let us suppose that the inner function used in the tree-based hash function is $f'$. 
According to Bertoni \emph{et al.} \cite{BDPV14_Suf}, prepending two bits to the intputs to $f'$ is sufficient to ensure the indifferentiability of 
the resulting tree-based hash function. 
The values of these bits depend on the type of $f'$-input, \emph{i.e.} the location of the input in the tree topology. In this paper, we choose to use
$N_b-1$ bits to encode the type of $f'$-input, where only two bits can be non-zero. For instance the binary encoding can be $b_0b_10^{N_b-3}$, where the values of $b_0$
and $b_1$ depend on the type of $f'$-input. We can remark that considering the prefix-free encoding and this second encoding, the first bit of the message is at the 
end of the second block. Our argument is that we can precompute all the possible hash states that result from the processing of the second block. If the number of 
possible input sizes (before padding or any prepending) is $k$ then the number of precomputed hash states is exactly $8k$. Thus, with these precomputed values,
the running time to process with $f'$ an input that can fit into $s$ blocks is exactly equal to $s$ units of time.~\\

\paragraph{\textbf{Inner function based on a sponge construction}} We could use a hash function like Keccak \cite{BDPV13_keccak} which does not require to embed 
the message size in the input.
This hash function, constructed on top of a permutation, uses a padding $10^*1$ at the end of the message so that 
the message size in bits corresponds to a multiple of the block size. 
More precisely, the appending consists of the bit 1, followed by the minimum number (possibly 0) of bits 0, followed by a bit 1. As in the previous solution,
the computation of the nodes requires to format appropriately the inputs of this function, using the necessary encoding (at least two bits) 
for sound tree hashing \cite{BDPV14_Suf}. The trick is to prepend to the input an encoding consisting of $N_b-2$ bits, where only $2$ bits can be non-zero 
(the values of these bits depend
on the type of input). We thus observe that the two first bits of the message are at the end of the first block. We can precompute all the possible hash states resulting from
the processing of this first block. There are four possibilities for the two first bits of the message, and four possibilities for the choice of the encoding, 
for a total of sixteen possible hash states. ~\\

\paragraph{\textbf{Inner function based on the compression function of Skein}} Skein \cite{FLSWBKCW09} uses a variable-input-length compression function which requires $l$ invocations to
the tweakable block cipher \textit{threefish} to compress a message of $l$ blocks. This ideal rate is due to the fact that the message is not padded 
when it is already a multiple of the block size.
The information of the lack of padding is included in the tweak, thus providing the same functionality as reversible padding. 
This compression function is indifferentiable of a random oracle if \textit{Threefish} 
acts as an ideal cipher.
As in the case before, we just have to prepend to the message an encoding of $N_b$ bits in order to distinguish $4$ types of input to the compression function.
In fact, only two are sufficient and the remaining bits serves only to reduce the number of hash states to precompute, \emph{i.e.} $4$ hash states.

\section{Optimal trees having their leaves at all the levels}\label{kangaroo}

The idea of processing both message blocks and chaining values (non-leaf nodes, \emph{i.e.} digests) using a single inner function evaluation 
was suggested in \cite{BDPV14_Sak} (under the name of \emph{kangaroo hopping}) in order to avoid certain computation overheads. 
With our assumptions, we first apply this idea for all nodes of a tree of restricted height with the aim of optimizing the parallel running time, 
and then we apply it to the case of trees of unrestricted height.

In the following results, the considered inner function has an idealized running time and is devoid of the padding overhead. 
If the padding is not neglected, the number of processors is underestimated: 
\begin{itemize}
 \item If the inner function is the sponge-based construction defined above, the number of involved processors should be multiplied by four, 
 because two bits have to be guessed at each parallel step.
 \item If the inner function is the prefix-free MD construction defined above, the number of involved processors should be multiplied by two, 
 because one bit has to be guessed at each parallel step.
\end{itemize}
For the sake of simplification, we first focus on the case $d=1$. The case $d>1$ is discussed in the subsequent subsection.

\subsection{Case $d=1$ (or $N_b=N_o$)}

\paragraph{\textbf{An algorithm for a tree of height 2}} Let us consider a message of size $l$, whose blocks are denoted $m_1$, $m_2$, $\ldots$, $m_l$. 
The processors are indexed $P_i$ with $i \geq 1$, and we make the assumption that they start their computations at the same time. The message is subdivided in chunks of 
variable size to be distributed to each processor:
\begin{itemize}
 \item $P_1$ and $P_2$ each receives a chunk of $2$ blocks and applies the inner function on these chunks.
 $P_2$ computes the hash of $m_3\|m_4$, while $P_1$ computes the hash of $m_1\|m_2$ without finalizing it. In other words, $P_1$ prepares to receive further consecutive 
 blocks. We denote by $c_2$ the digest computed by~$P_2$.
 \item As long as there remains message blocks, $P_i$ (with $i \geq 3$) receives $i$ blocks and applies the inner function on their concatenation. 
 We denote by $c_i$ the resulted digest computed by $P_i$ for $i \geq 3$. Note that $P_i$ can possibly process less than $i$ blocks if the end of the message is reached.
 \item $P_1$ continues to evaluate the inner function on the collected digests $c_2$, $\ldots$, $c_k$ as they arrive. The evaluation of the inner function is resumed 
 immediately when a digest $c_i$ is available.
 \item Assuming that $P_k$ is the last processor that has received blocks, the final digest computed by $P_1$ corresponds to the evaluation of the inner function on 
 $$m_1 \| m_2 \| c_2 \| c_3 \| \cdots \| c_{k-1} \| c_k.$$
\end{itemize}

 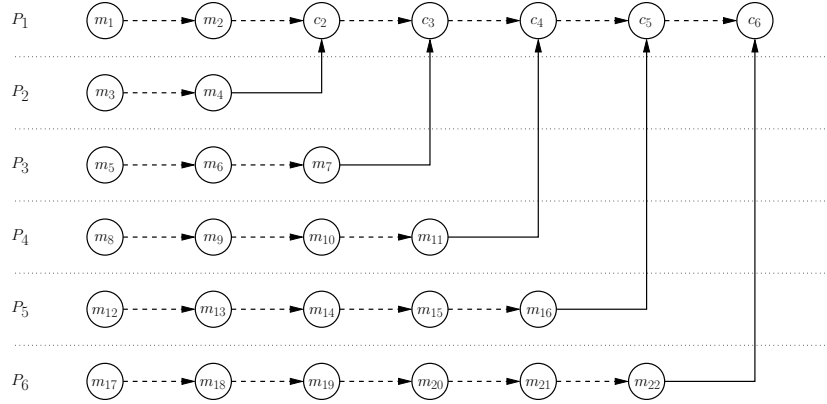
\begin{figure}[!h]
\centering
\scalebox{0.24}{
 \input{optimal_two_level_tree.pspdftex}
 }
 \caption{Processing of a message of $22$ blocks using $8$ processors, denoted $P_1$, $P_2$, ..., $P_6$. We can see that $P_1$ and $P_2$ each process $2$
 blocks of the message, while $P_i$, with $i \geq 3$, processes $i$ blocks of the message. The chaining values $c_2$, $c_3$, ..., $c_6$ 
 are collected and processed by $P_1$ as soon as they are computed.}
 \label{optimal_two_level_tree}
\end{figure}

An example of execution of this algorithm is depicted in Figure \ref{optimal_two_level_tree}. 
The running time of this hash function is the running time for computing $c_k$ plus one, \emph{i.e.}, $k+1$ units of time. 
If the last processor receives a single block, this one can be processed by the first processor in order to save one processor, while leaving 
unchanged the running time of $k+1$. Note that $k$ is such that $\sum_{i=1}^k i \geq l-1$, \emph{i.e.} such that $k^2+k-2(l-1) \geq 0$. This inequation has 
two solutions $\frac{-1 \pm \sqrt{8l-7}}{2}$, of which only one is positive for $l \geq 1$. The solution 
is then $k = \left\lceil \frac{-1 + \sqrt{8l-7}}{2} \right\rceil$. 
Among the tree structures of height 2, the one used in the algorithm above leads to an optimal
parallel running time. While conserving this running time, one may
desire to decrease the number of involved processors.

\begin{theorem}
 Let a message of length $l$ blocks such that $l \geq 2$. We can construct a hash tree of height 2 
 allowing a parallel running time of $k+1$ units of time, using 
 $k-i+2$ processors, where 
 $$k=\left\lceil \frac{-1+\sqrt{4i^2-12i+8l+1}}{2} \right\rceil = \left\lceil \frac{-1+\sqrt{8l-8}}{2} \right\rceil$$
 and
 $$i=\max_j\ \mathrm{argmin}_j \left\lceil\frac{-1+\sqrt{4j^2-12j+8l+1}}{2}\right\rceil<\frac{\sqrt{4+4\sqrt{8l-8}}+3}{2}.$$
\end{theorem}

\begin{proof}
 We seek to maximize $i$ and minimize $k$ such that Processor $P_1$ processes $i$ message blocks and $k-i+1$ chaining values, and
 $P_2$, $P_3$, ..., $P_{k-i+1}$, $P_{k-i+2}$ process respectively $i$, $i+1$, ..., $k-1$, $k$ message blocks. The final digest computed by $P_1$
 corresponds to the evaluation of the inner function on
 $$m_1\|m_2\|\cdots \| m_i \| c_i \|c_{i+1} \| \cdots \| c_{k-1}\|c_{k}.$$
 Thus, we seek to minimize $k$ such that $\sum_{j=i-1}^{k}j \geq l-1$. We have to solve the inequation $k^2-k-i^2+3i-2l \geq 0$. The discriminant 
 $\Delta=4i^2-12i+8l+1$ is strictly positive iff $i^2 - 3i + 2l +1/4 \geq 0$. Since this last inequality is verified for $n \geq i \geq 1$, we have two
 solutions of which only one is positive. We deduce that $k= \left\lceil \frac{-1 + \sqrt{4i^2-12i+8l+1}}{2} \right\rceil$
 for $n \geq i \geq 1$. Let us consider the function $f(x)=\frac{-1+\sqrt{4x^2-12x+8l+1}}{2}$, whose derivative is $f'(x)=\frac{2x-3}{\sqrt{8l-4x^2-12x+1}}$.
 Solving the equation $f'(x)=0$ leads to the solution $x=3/2$. Since $f(x)$ is increasing for $x > 3/2$, we now have to seek the maximum integer $x$ satisfying
 $\left\lceil \frac{-1 + \sqrt{4x^2-12x+8l+1}}{2} \right\rceil=\lceil f(3/2) \rceil$. We can upper bound $x$ such that $f(x) < f(3/2) + 1$,
 leading to the expected result.
\end{proof}
~\\

According to \cite{AB16}, for a tree of height $k$, 
the optimal parallel running time is in $O(l^{\frac{1}{k}})$, where $l$ is the size of the message. 
We recall that this result was shown for both the hashing of stored content (the size of the message 
has to be known in advance) and the hashing of live-streamed content.
The construction above, which supports the processing of live-streamed content, does not contradict this result.
Anyway, we recall that in our settings, the message is supposed to be already available, and thus the need of the message size as input to the algorithm 
does not matter.
We see here that the optimization of such a tree using both \emph{kangaroo hopping} and increasing input sizes is interesting.
One advantage of using a tree of restricted height is its limited memory usage in a sequential execution of the algorithm. 
If memory usage for a sequential execution is not a concern, we can consider trees of unrestricted height.

\begin{theorem}\label{Theorem_opt}
 Let a message of length $l$ blocks. We can construct a hash tree allowing a parallel running time of exactly $\lceil \log_2 l \rceil+1$ units of time, using 
 $\lceil l/2 \rceil$ processors.
\end{theorem}

\begin{proof}
We first give the construction of a tree structure. Then, we consider a hash function based on it, and we give a scheduling strategy to perform all 
the computations in parallel.
Let us consider a binary tree of height $h=\lceil \log_2 l \rceil$. We denote by $l_i$ the number of nodes of level $i \geq 1$. Remark that this 
binary tree can be such that all its leaves are at the same depth and $l_i=\lceil l/2^i \rceil$. 
The $l_i$ nodes of the level $i$ are indexed. For $j=1 \ldots l_i$, one node of this level is denoted $N_j$ and, 
in particular, its leftmost child is denoted $N_{j,LC}$. 
\emph{Note that if $N_j$ has a single child, this latter is still considered as its leftmost child.}
At each level $i$ of this tree, 
starting from level $2$ up to level $h$, we transform the nodes in the following way: 
for $j=1 \ldots l_i$, the node $N_{j,LC}$ is discarded and 
its children become the children of $N_j$.
\emph{We notice that once this operation is performed, a node $N_j$ can have a higher number of children.}
The result is a tranformed tree which is no longer a binary tree and where leaves are located at all the levels. 
An example of execution of this algorithm is depicted in Figure \ref{Derivation_tree}.

We now consider a hash function based on this tree structure. The computations are done in parallel in the following way:
in a same parallel step, each processor starts the computation of one of the $\lceil l/2 \rceil$ nodes that has leaves. This parallel step requires 2 units of time.
Hence, the computations of these nodes (or of their parent nodes) can progress in a parallel step of one unit of time. We need to repeat such a parallel step 
as many times as necessary to complete the processing of this hash tree, \emph{i.e.}, $\lceil \log_2 l \rceil-1$ times. We then deduce a parallel running time
of $\lceil \log_2 l \rceil-1 +2$ units of time. The number of involved processors corresponds to the number of nodes having leaves, \emph{i.e.} $\lceil l/2 \rceil$.
An example of parallel hash computation is depicted in Figure \ref{another_tree_representation}.
\end{proof}


\begin{figure}[h]
  \begin{center}
    \subfloat[Binary tree]{
      \includegraphics[height=2.4cm]{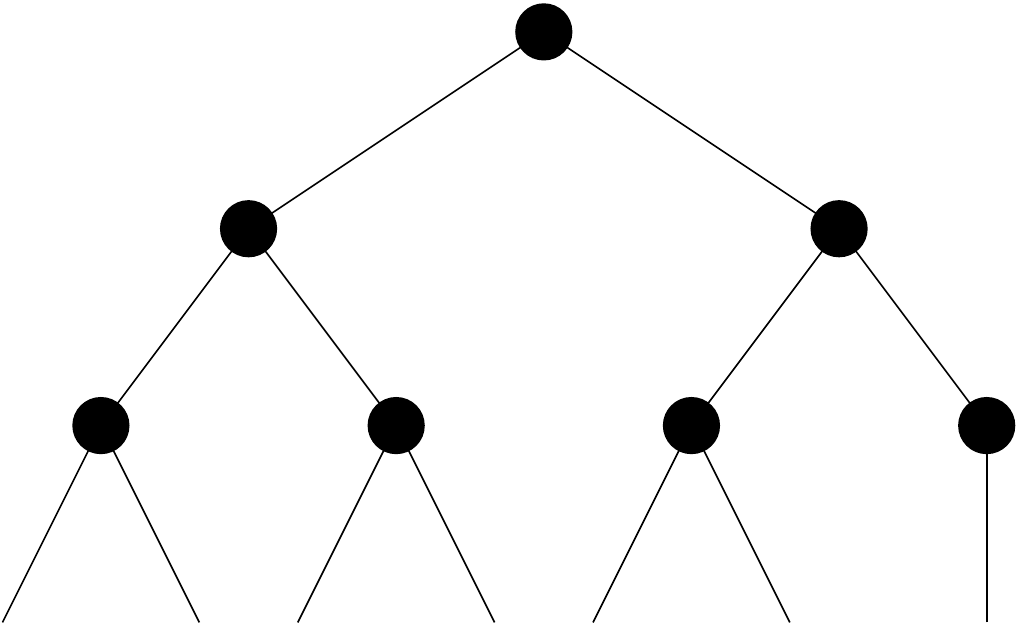}
      \label{sub:init}
                         }
    ~~~~~~~~~~~~~~~~~~~~~~
    \subfloat[Result of the first iteration]{
      \includegraphics[height=2.4cm]{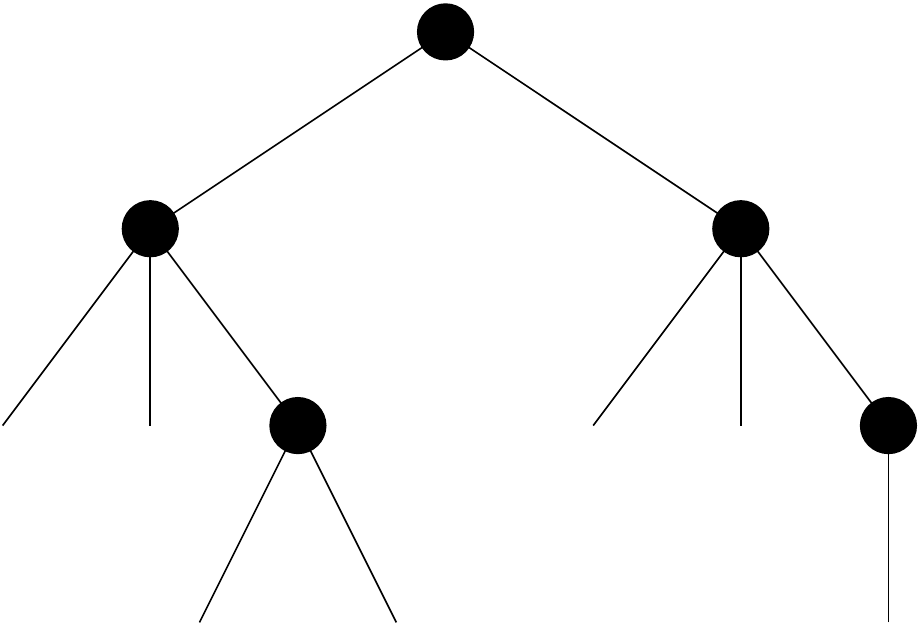}
      \label{sub:first}
                         }

    \subfloat[Result of the second (and last) iteration]{
      \includegraphics[height=2.4cm]{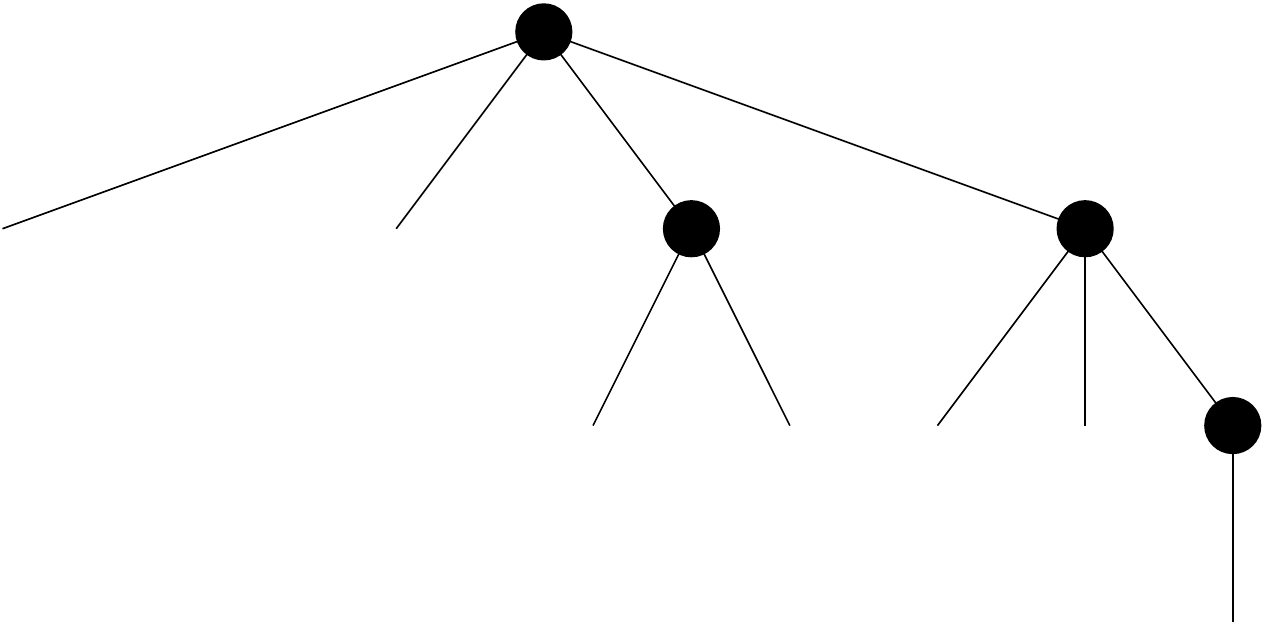}
      \label{sub:second}
                         }
    \caption{Derivation of a tree structure having its leaves at all the levels from a classic binary tree that processes a message of 7 blocks}
    \label{Derivation_tree}
  \end{center}
\end{figure}

~\\

\paragraph{\textbf{Remark}} Assuming a message of length $l$ blocks, the parallel running time of $\lceil \log_2 l \rceil+1$ is optimal. 
Indeed, this is clearly true for a message of 4 blocks which requires 3 units of time.
 Let us suppose that the running time of $k + 1$ is optimal for a message of length $2^k$. If we cannot process more than $2^k$ blocks in $k+1$ units of time, processing
 $2^k$ more blocks requires at least one more unit of time.
 Thus, the running time of $k+2$ is still optimal for a message of length $2^{k+1}$.~\\

 \begin{figure}[!t]
\centering
\scalebox{0.55}{
 \input{tree_other_representation.pspdftex}
 }
 \caption{Example of the processing of 8 blocks $m_1$, $m_2$, ..., $m_8$, using another tree representation. 
 We have 4 processors denoted $P_1$, $P_2$, $P_3$ and $P_4$. A dotted line
 represents a serial computation using the same hash context, while a solid line indicates that a hash state is used by another hash context. The encircled message blocks
 or chaining values that are connected with a dotted line are in the same $f$-input. For instance,
 the processor $P_2$ computes the hash of $m_3 \| m_4$, denoted $c_1$. The chaining values $c_1$ and $c_3$ are used by the hash context of the processor $P_1$.
 The parallel running time to compute the root node is equal to the running time required for computing the hash of $m_1\|m_2\|c_1\|c_3$, \emph{i.e.}, 4 units of time.}
 \label{another_tree_representation}
\end{figure}
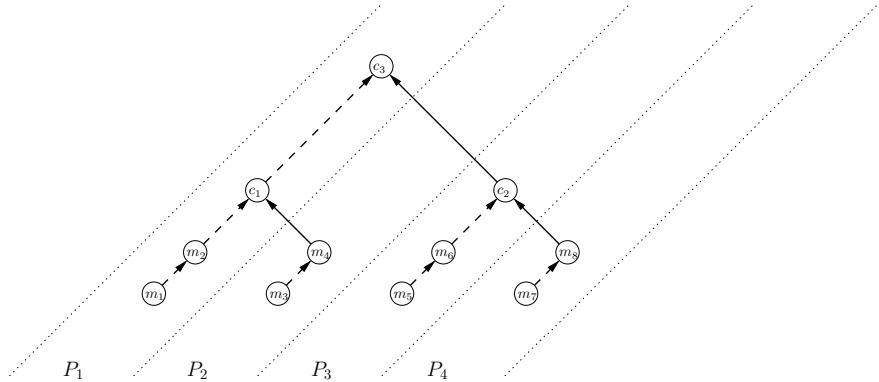


\paragraph{\textbf{Performances improvements}} The (parallel) running time of such a tree is to be compared with the running time of
an optimal tree having its leaves at the same depth \cite{AR15}, \emph{i.e.} approximately $3 \lceil \log_3 l \rceil$. This represents, approximately,
a {\texttt 2x} speedup.~\\


\textbf{What if we apply the algorithm above on a ternary tree to construct another tree?} The transformed tree would lead to a parallel running time of at most
$(\lceil \log_3 l \rceil-1) \cdot 2 + 3 = 2\lceil \log_3 l \rceil+1$. More precisely, if the ternary tree has a root node of arity 3, then the 
hash function based on the transformed tree has a parallel running time of exactly $2\lceil \log_3 l \rceil+1$. Otherwise, if it is of arity 2, the
transformed tree leads to a parallel running time of exactly $2\lceil \log_3 l \rceil$.

For a large message length $l$, we have  $\lceil \log_2 l \rceil + 1 < 2\lceil \log_3 l \rceil$. It is thus more interesting to use the topology derived from a 
binary tree. For a finite and small number of $l$, the tree topology derived from a ternary tree gives the same running time. For these message lengths,
such a topology is preferable since it decreases the number of involved processors. For the reasons outlined below, deriving a topology from a quaternary tree or 
any tree of arity greater than 4 worsen the parallel running time.~\\


\textbf{Can we further decrease the number of processors while conserving the running time stated in the theorem above?} To do so, we should be able to increase the number of 
nodes or message blocks processed by one processor during one parallel step. Let us see a counter-example. Suppose that we 
have a hash tree that can be processed in a parallel running time of $\lceil \log_2 l \rceil+1$ and that one node in this tree has more than 3 leaves, say $x$ leaves.
We have $\lceil \log_2 x \rceil+1 < x$ when $x > 3$, meaning that we can transform this node in order to improve the overall running time.

%

\begin{theorem}
 Let $l$ be the number of blocks of the message and let $i$ be the biggest integer such that $2^i < l$. 
 The optimal parallel running time can be reached using only $\lceil l/3 \rceil$ processors 
 if $2^i<l\leq3 \cdot 2^{i-1}$ and $\lceil l/2 \rceil$ processors if $3 \cdot 2^{i-1}< l \leq 2^{i+1}$.
\end{theorem}

\begin{proof}
 We just allow the derived tree (in the proof above) to have $3$ leaves per node, instead of $2$. We recall that the parallel running time 
 is $\lceil \log_2(l/2) \rceil+2$ with $2$ leaves per node, whereas it is $\lceil \log_2(l/3) \rceil+3$ using $3$ leaves per node. If 
 we have $\lceil \log_2(l/3) \rceil+3 \leq \lceil \log_2(l/2) \rceil+2$
 for a given $l$, the second tree structure should be used to decrease the number of processors to $\lceil l/3 \rceil$.  
 We now determine the range of values of $l$ for which $\lceil \log_2(l/3) \rceil+2 \leq \lceil \log_2 l \rceil$.
 Let us set $u=\log_2 l$. We rewrite the inequality 
 $\lceil \log_2(l/3) \rceil+2 \leq \lceil \log_2 l \rceil$
 as  
 \begin{equation}\label{ineq_lab}
  \lceil i + f -\log_2(3) +2 \rceil \leq \lceil i + f \rceil
 \end{equation}
 where $f$ is the fractional part of $u$ and $i$ is its integer part. Since $2-\log_2(3)>0$, $f$ is necessarily non-zero.
 Thus, Inequality (\ref{ineq_lab}) is satisfied iff $0 < f \leq \log_2 3 - 1$, \emph{i.e.} iff $i < \log_2 l \leq i + \log_2(3/2)$. This
 leads to the expected intervals of validity.
 
\end{proof}

Another question is the optimal parallel running time that can be obtained using a fixed number of processors.

\begin{theorem}
 Let $l$ be the number of blocks of the message and let $P$ be the number of processors. 
 There exists a mode having an optimal parallel running time of $\lceil l/P \rceil$ + $\lceil \log_2 P \rceil$ units of time.
\end{theorem}

\begin{proof}
 Let us consider a message of $2P$ blocks. According to Theorem \ref{Theorem_opt}, it can be hashed in a parallel running time 
 of $\lceil \log_2 2P \rceil+1$ ($=\lceil \log_2 P \rceil+2$) units of time. During the first two units of time, each processor processes $2$ blocks of the message.
 We thus replace these $2$ blocks by at most $\lceil l/P \rceil$ blocks that can be hashed sequentially in at most $\lceil l/P \rceil$ units of time.
 Since there always exists two integers $a \geq 0$ and $b \geq 0$ such that $a + b = P$ and
$a\lceil l/P \rceil + b \lfloor l/P \rfloor = l$, we conclude the result.
\end{proof}

\subsection{Case $d>1$}

\begin{theorem}\label{Theorem_opt}
 Let a message of length $l$ blocks. We can construct a hash tree allowing a parallel running time of exactly $\lceil \log_{(d+1)} (l/2) \rceil+2$ 
 units of time, using $\lceil l/2 \rceil$ processors.
\end{theorem}

\begin{proof}
 First, we observe that $2(d+1)$ blocks of the message can be compressed in $3$ units of times, using $d+1$ processors. 
 Indeed, $d+1$ processors can each compress 2 blocks, and the first one can continue the evaluation of its hash function by processing the chaining values 
 produced by the $d$ other processors.
 Given a hash state,
 we can compress $d$ subsequent chaining values in one unit of time. Thus, we can compress $d+1$ times more blocks (\emph{i.e.} $3(d+1)^2$ in total)
 in one more unit of time, and by using $d+1$ times more processors. Repeating this recursively, we obtain a single chaining value (the root node) at an iteration $k$. 
 It appears that $k$ is the smallest integer satisfying the inequality $2(d+1)^k \geq l$. The total parallel running time then corresponds to the time required by 
 the most loaded processor: the running time to process two blocks of the message, in addition to the running time to process at most $dk$ chaining values,
 \emph{i.e.} $k$ units of time, yielding the expected result.
\end{proof}
~\\
Note that $\forall d \geq 1 \forall x > 3$, we have $\lceil \log_{d+1}(x) - \log_{d+1}(2) \rceil + 2 < x$, meaning that more than 3 leaves per node lead to a suboptimal
parallel time.

\begin{theorem}
 Let $l$ be the number of blocks of the message and let $i$ be the biggest integer such that $2(d+1)^i < l$. 
 The optimal parallel running time can be reached using only $\lceil l/3 \rceil$ processors 
 if $2(d+1)^i<l\leq 3(d+1)^i$ and $\lceil l/2 \rceil$ processors if $3(d+1)^i< l \leq 2(d+1)^{i+1}$.
\end{theorem}

\begin{proof}
 A tree with 3 leaves per node is preferable if $\lceil \log_{(d+1)}(l/3) \rceil + 3 \leq \lceil \log_{(d+1)}(l/2) \rceil + 2$.
 We have to determine the range of values $l$ which fulfil this inequality.
 Let us set $u=\log_{(d+1)} (l/2)$. We rewrite the inequality 
 $\lceil \log_{(d+1)}(l/3) \rceil+1 \leq \lceil \log_{(d+1)} (l/2) \rceil$
 as  
 \begin{equation}\label{ineq_lab}
  \lceil i + f +\log_{(d+1)}(2/3) + 1 \rceil \leq \lceil i + f \rceil
 \end{equation}
 where $f$ and $i$ are respectively the fractional part and the integer part of $u$. Since $\log_{(d+1)}(2/3) + 1>0$ for all $d \geq 1$, $f$ is necessarily non-zero.
 Thus, Inequality~(\ref{ineq_lab}) is satisfied iff $0 < f \leq -\log_{(d+1)}(2/3)$, \emph{i.e.} iff $i < \log_{(d+1)} (l/2) \leq i - \log_{(d+1)}(2/3)$. This
 leads to the expected intervals of validity.
 
\end{proof}

\begin{theorem}
 Let $l$ be the number of blocks of the message and let $P$ be the number of processors. 
 There exists a mode having an optimal parallel running time of at most $\lceil l/P \rceil$ + $\lceil \log_{(d+1)} (P) \rceil$ units of time.
\end{theorem}

\begin{proof}
This theorem follows immediately from the previous one. We replace the size $l$ by $2P$. This message is then compressed in
$2 + \lceil \log_{(d+1)} P \rceil$ units of time. In this scheme, each processor starts by compressing two blocks of the message. If we replace these $2$
blocks by $\lceil l/P \rceil$ blocks, this means that we can compress a message of at most $\lceil \frac{l}{P} \rceil P$ blocks 
in $\lceil l/P \rceil + \lceil \log_{(d+1)} P \rceil$ units of time. Since there always exists integers $a \geq 0$ and $b \geq 0$ such that $a + b = P$ and
$a\lceil l/P \rceil + b \lfloor l/P \rfloor = l$, we conclude the result.
\end{proof}

\bibliographystyle{plain}
\bibliography{trees}

\end{document}

%% file: optimal_two_level_tree.pspdftex
\begin{picture}(0,0)%
\includegraphics{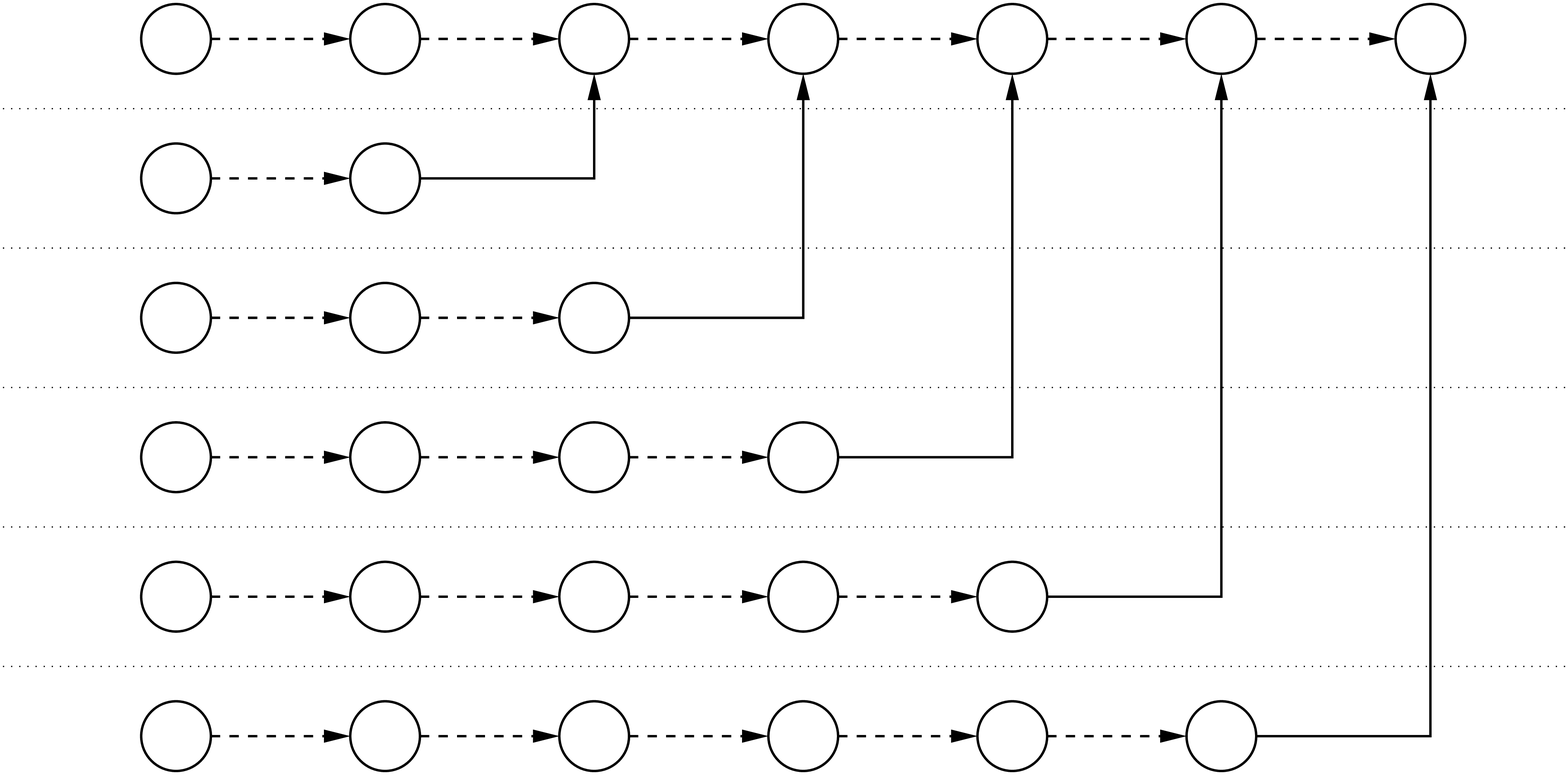}%
\end{picture}%
\setlength{\unitlength}{4144sp}%
\begingroup\makeatletter\ifx\SetFigFont\undefined%
\gdef\SetFigFont#1#2#3#4#5{%
  \reset@font\fontsize{#1}{#2pt}%
  \fontfamily{#3}\fontseries{#4}\fontshape{#5}%
  \selectfont}%
\fi\endgroup%
\begin{picture}(20371,10153)(1252,-12233)
\put(3337,-2839){\makebox(0,0)[lb]{\smash{{\SetFigFont{29}{34.8}{\familydefault}{\mddefault}{\updefault}{\color[rgb]{0,0,0}$m_1$}%
}}}}
\put(6056,-2839){\makebox(0,0)[lb]{\smash{{\SetFigFont{29}{34.8}{\familydefault}{\mddefault}{\updefault}{\color[rgb]{0,0,0}$m_2$}%
}}}}
\put(8802,-2839){\makebox(0,0)[lb]{\smash{{\SetFigFont{29}{34.8}{\familydefault}{\mddefault}{\updefault}{\color[rgb]{0,0,0}$c_2$}%
}}}}
\put(11496,-2839){\makebox(0,0)[lb]{\smash{{\SetFigFont{29}{34.8}{\familydefault}{\mddefault}{\updefault}{\color[rgb]{0,0,0}$c_3$}%
}}}}
\put(14203,-2839){\makebox(0,0)[lb]{\smash{{\SetFigFont{29}{34.8}{\familydefault}{\mddefault}{\updefault}{\color[rgb]{0,0,0}$c_4$}%
}}}}
\put(16911,-2839){\makebox(0,0)[lb]{\smash{{\SetFigFont{29}{34.8}{\familydefault}{\mddefault}{\updefault}{\color[rgb]{0,0,0}$c_5$}%
}}}}
\put(19631,-2839){\makebox(0,0)[lb]{\smash{{\SetFigFont{29}{34.8}{\familydefault}{\mddefault}{\updefault}{\color[rgb]{0,0,0}$c_6$}%
}}}}
\put(3361,-4644){\makebox(0,0)[lb]{\smash{{\SetFigFont{29}{34.8}{\familydefault}{\mddefault}{\updefault}{\color[rgb]{0,0,0}$m_3$}%
}}}}
\put(6056,-4644){\makebox(0,0)[lb]{\smash{{\SetFigFont{29}{34.8}{\familydefault}{\mddefault}{\updefault}{\color[rgb]{0,0,0}$m_4$}%
}}}}
\put(3349,-6462){\makebox(0,0)[lb]{\smash{{\SetFigFont{29}{34.8}{\familydefault}{\mddefault}{\updefault}{\color[rgb]{0,0,0}$m_5$}%
}}}}
\put(6056,-6449){\makebox(0,0)[lb]{\smash{{\SetFigFont{29}{34.8}{\familydefault}{\mddefault}{\updefault}{\color[rgb]{0,0,0}$m_6$}%
}}}}
\put(8764,-6449){\makebox(0,0)[lb]{\smash{{\SetFigFont{29}{34.8}{\familydefault}{\mddefault}{\updefault}{\color[rgb]{0,0,0}$m_7$}%
}}}}
\put(3335,-8267){\makebox(0,0)[lb]{\smash{{\SetFigFont{29}{34.8}{\familydefault}{\mddefault}{\updefault}{\color[rgb]{0,0,0}$m_8$}%
}}}}
\put(6056,-8254){\makebox(0,0)[lb]{\smash{{\SetFigFont{29}{34.8}{\familydefault}{\mddefault}{\updefault}{\color[rgb]{0,0,0}$m_9$}%
}}}}
\put(8688,-8254){\makebox(0,0)[lb]{\smash{{\SetFigFont{29}{34.8}{\familydefault}{\mddefault}{\updefault}{\color[rgb]{0,0,0}$m_{10}$}%
}}}}
\put(11394,-8254){\makebox(0,0)[lb]{\smash{{\SetFigFont{29}{34.8}{\familydefault}{\mddefault}{\updefault}{\color[rgb]{0,0,0}$m_{11}$}%
}}}}
\put(3271,-10072){\makebox(0,0)[lb]{\smash{{\SetFigFont{29}{34.8}{\familydefault}{\mddefault}{\updefault}{\color[rgb]{0,0,0}$m_{12}$}%
}}}}
\put(5968,-10059){\makebox(0,0)[lb]{\smash{{\SetFigFont{29}{34.8}{\familydefault}{\mddefault}{\updefault}{\color[rgb]{0,0,0}$m_{13}$}%
}}}}
\put(8687,-10071){\makebox(0,0)[lb]{\smash{{\SetFigFont{29}{34.8}{\familydefault}{\mddefault}{\updefault}{\color[rgb]{0,0,0}$m_{14}$}%
}}}}
\put(11383,-10071){\makebox(0,0)[lb]{\smash{{\SetFigFont{29}{34.8}{\familydefault}{\mddefault}{\updefault}{\color[rgb]{0,0,0}$m_{15}$}%
}}}}
\put(14103,-10071){\makebox(0,0)[lb]{\smash{{\SetFigFont{29}{34.8}{\familydefault}{\mddefault}{\updefault}{\color[rgb]{0,0,0}$m_{16}$}%
}}}}
\put(3260,-11877){\makebox(0,0)[lb]{\smash{{\SetFigFont{29}{34.8}{\familydefault}{\mddefault}{\updefault}{\color[rgb]{0,0,0}$m_{17}$}%
}}}}
\put(5967,-11864){\makebox(0,0)[lb]{\smash{{\SetFigFont{29}{34.8}{\familydefault}{\mddefault}{\updefault}{\color[rgb]{0,0,0}$m_{18}$}%
}}}}
\put(8675,-11876){\makebox(0,0)[lb]{\smash{{\SetFigFont{29}{34.8}{\familydefault}{\mddefault}{\updefault}{\color[rgb]{0,0,0}$m_{19}$}%
}}}}
\put(11395,-11864){\makebox(0,0)[lb]{\smash{{\SetFigFont{29}{34.8}{\familydefault}{\mddefault}{\updefault}{\color[rgb]{0,0,0}$m_{20}$}%
}}}}
\put(14078,-11864){\makebox(0,0)[lb]{\smash{{\SetFigFont{29}{34.8}{\familydefault}{\mddefault}{\updefault}{\color[rgb]{0,0,0}$m_{21}$}%
}}}}
\put(16799,-11864){\makebox(0,0)[lb]{\smash{{\SetFigFont{29}{34.8}{\familydefault}{\mddefault}{\updefault}{\color[rgb]{0,0,0}$m_{22}$}%
}}}}
\put(1267,-6439){\makebox(0,0)[lb]{\smash{{\SetFigFont{65}{78.0}{\familydefault}{\mddefault}{\updefault}{\color[rgb]{0,0,0}$P_3$}%
}}}}
\put(1267,-8239){\makebox(0,0)[lb]{\smash{{\SetFigFont{65}{78.0}{\familydefault}{\mddefault}{\updefault}{\color[rgb]{0,0,0}$P_4$}%
}}}}
\put(1267,-10039){\makebox(0,0)[lb]{\smash{{\SetFigFont{65}{78.0}{\familydefault}{\mddefault}{\updefault}{\color[rgb]{0,0,0}$P_5$}%
}}}}
\put(1267,-11854){\makebox(0,0)[lb]{\smash{{\SetFigFont{65}{78.0}{\familydefault}{\mddefault}{\updefault}{\color[rgb]{0,0,0}$P_6$}%
}}}}
\put(1267,-2839){\makebox(0,0)[lb]{\smash{{\SetFigFont{65}{78.0}{\familydefault}{\mddefault}{\updefault}{\color[rgb]{0,0,0}$P_1$}%
}}}}
\put(1267,-4639){\makebox(0,0)[lb]{\smash{{\SetFigFont{65}{78.0}{\familydefault}{\mddefault}{\updefault}{\color[rgb]{0,0,0}$P_2$}%
}}}}
\end{picture}%

%% file: tree_other_representation.pspdftex
\begin{picture}(0,0)%
\includegraphics{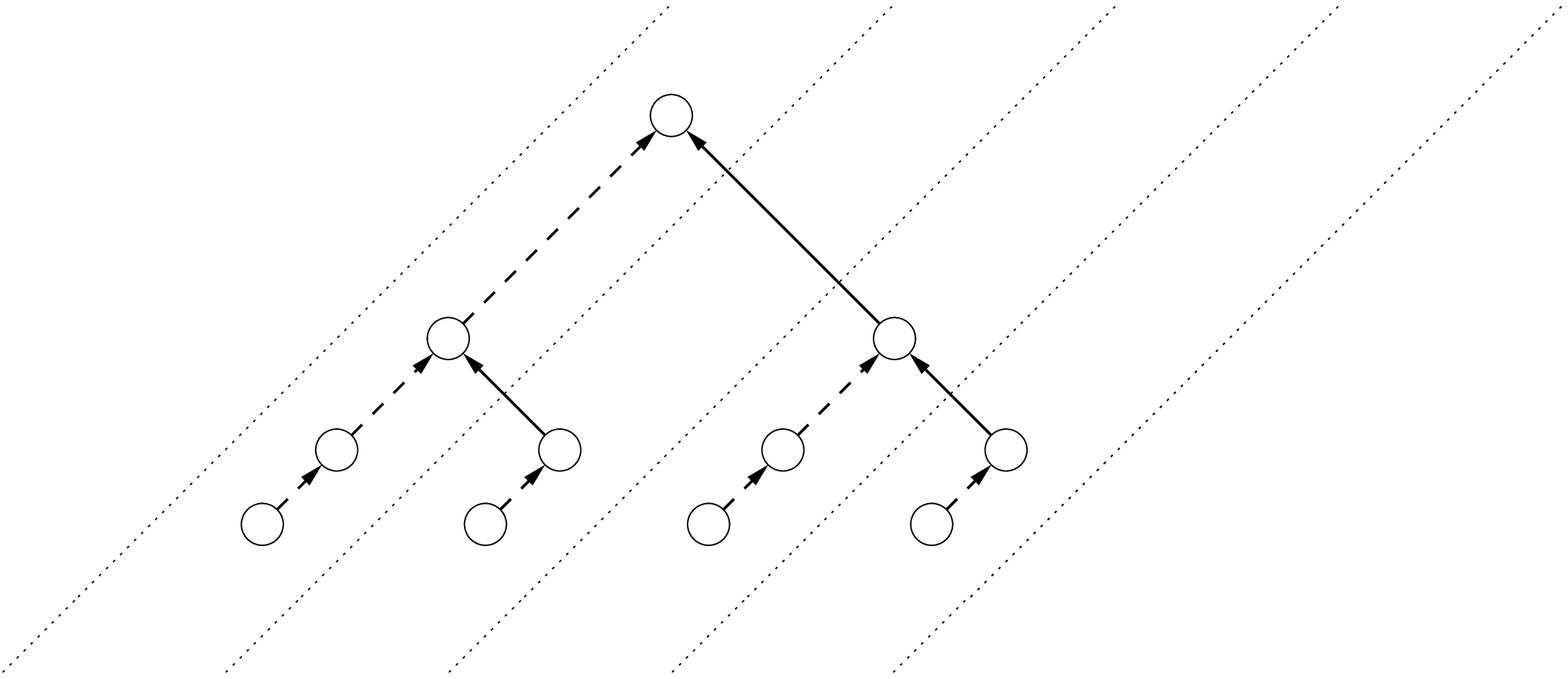}%
\end{picture}%
\setlength{\unitlength}{4144sp}%
\begingroup\makeatletter\ifx\SetFigFont\undefined%
\gdef\SetFigFont#1#2#3#4#5{%
  \reset@font\fontsize{#1}{#2pt}%
  \fontfamily{#3}\fontseries{#4}\fontshape{#5}%
  \selectfont}%
\fi\endgroup%
\begin{picture}(9474,4144)(5164,-9593)
\put(5761,-9511){\makebox(0,0)[lb]{\smash{{\SetFigFont{14}{16.8}{\familydefault}{\mddefault}{\updefault}{\color[rgb]{0,0,0}$P_1$}%
}}}}
\put(7111,-9511){\makebox(0,0)[lb]{\smash{{\SetFigFont{14}{16.8}{\familydefault}{\mddefault}{\updefault}{\color[rgb]{0,0,0}$P_2$}%
}}}}
\put(8461,-9511){\makebox(0,0)[lb]{\smash{{\SetFigFont{14}{16.8}{\familydefault}{\mddefault}{\updefault}{\color[rgb]{0,0,0}$P_3$}%
}}}}
\put(9721,-9511){\makebox(0,0)[lb]{\smash{{\SetFigFont{14}{16.8}{\familydefault}{\mddefault}{\updefault}{\color[rgb]{0,0,0}$P_4$}%
}}}}
\put(6661,-8656){\makebox(0,0)[lb]{\smash{{\SetFigFont{10}{12.0}{\familydefault}{\mddefault}{\updefault}{\color[rgb]{0,0,0}$m_1$}%
}}}}
\put(8011,-8656){\makebox(0,0)[lb]{\smash{{\SetFigFont{10}{12.0}{\familydefault}{\mddefault}{\updefault}{\color[rgb]{0,0,0}$m_3$}%
}}}}
\put(9361,-8656){\makebox(0,0)[lb]{\smash{{\SetFigFont{10}{12.0}{\familydefault}{\mddefault}{\updefault}{\color[rgb]{0,0,0}$m_5$}%
}}}}
\put(10711,-8656){\makebox(0,0)[lb]{\smash{{\SetFigFont{10}{12.0}{\familydefault}{\mddefault}{\updefault}{\color[rgb]{0,0,0}$m_7$}%
}}}}
\put(7111,-8206){\makebox(0,0)[lb]{\smash{{\SetFigFont{10}{12.0}{\familydefault}{\mddefault}{\updefault}{\color[rgb]{0,0,0}$m_2$}%
}}}}
\put(8461,-8206){\makebox(0,0)[lb]{\smash{{\SetFigFont{10}{12.0}{\familydefault}{\mddefault}{\updefault}{\color[rgb]{0,0,0}$m_4$}%
}}}}
\put(9811,-8206){\makebox(0,0)[lb]{\smash{{\SetFigFont{10}{12.0}{\familydefault}{\mddefault}{\updefault}{\color[rgb]{0,0,0}$m_6$}%
}}}}
\put(11161,-8206){\makebox(0,0)[lb]{\smash{{\SetFigFont{10}{12.0}{\familydefault}{\mddefault}{\updefault}{\color[rgb]{0,0,0}$m_8$}%
}}}}
\put(7786,-7531){\makebox(0,0)[lb]{\smash{{\SetFigFont{10}{12.0}{\familydefault}{\mddefault}{\updefault}{\color[rgb]{0,0,0}$c_1$}%
}}}}
\put(10486,-7531){\makebox(0,0)[lb]{\smash{{\SetFigFont{10}{12.0}{\familydefault}{\mddefault}{\updefault}{\color[rgb]{0,0,0}$c_2$}%
}}}}
\put(9136,-6181){\makebox(0,0)[lb]{\smash{{\SetFigFont{10}{12.0}{\familydefault}{\mddefault}{\updefault}{\color[rgb]{0,0,0}$c_3$}%
}}}}
\end{picture}%